\documentclass[runningheads,a4paper]{llncs}
\usepackage[utf8]{inputenc}
\usepackage[T1]{fontenc}
\usepackage{amssymb}
\usepackage{listings}
\usepackage{hyperref}
\usepackage{mathtools}
\usepackage{bbold}
\usepackage{algorithm}
\usepackage{xcolor}
\usepackage[most]{tcolorbox}
\usepackage{enumitem}
\usepackage{breakcites}

\usepackage[noend]{algpseudocode}

\usepackage[english]{babel}

\usepackage{graphicx}
\usepackage{amsfonts}

\newcommand{\sampledfrom}{\overset{{\scriptscriptstyle\$}}{\leftarrow}}

\newcommand{\Z}{\mathbb{Z}}

\newcommand{\N}{\mathbb{N}}

\newcommand{\F}{\mathbb{F}}

\newcommand{\poly}{\mathrm{\,poly}}

\newcommand{\inner}[1]{\langle#1\rangle}

\makeatletter
\newcommand{\mc}[1]{\mathcal{#1}}

\makeatother

\tcbset{
  protocolstyle/.style={
    colback=white,
    colframe=black,
    fonttitle=\bfseries,
    enhanced,
    left=2mm,
    right=2mm,
    top=2mm,
    bottom=2mm,
    boxrule=0.8pt,
    float,
    floatplacement=htbp, 
    breakable=false,
    arc=0mm,
  }
}
\newtcbtheorem{protocol}{Protocol}%
{protocolstyle}{prot}

\newlist{stepenumerate}{enumerate}{2}
\setlist[stepenumerate,1]{label=\arabic*.,ref=\arabic*}
\setlist[stepenumerate,2]{label=\arabic{stepenumeratei}.\arabic*.,ref=\arabic{stepenumeratei}.\arabic*}

\newcommand{\sparsedotcost}{\tilde{\mathfrak{d}}}

\newcommand{\matmulcost}{\mathfrak{m}}

\newcommand{\subspacedim}{n_1}\newcommand{\noiserate}{\mu}

\title{Practical Secure Delegated Linear Algebra with Trapdoored Matrices}
\author{Mark~Braverman\orcidID{0000-0003-1276-6081} \and Stephen~Newman\thanks{Corresponding author: \href{mailto:stephen.newman@princeton.edu}{stephen.newman@princeton.edu}}\orcidID{0000-0003-1948-9478}}

\index{Braverman, Mark}
\index{Newman, Stephen}
\authorrunning{M. Braverman and S. Newman}
\institute{Princeton University}
\begin{document}
\maketitle

\begin{abstract}
	Most heavy computation occurs on servers owned by a second party. This reduces data privacy, resulting in interest in data-oblivious computation, which typically severely degrades performance. Secure and fast delegated computation is particularly important for linear algebra, which comprises a large fraction of total computation and is best run on highly specialized hardware often accessible only through the cloud. 
    
    We state the natural efficiency and security desiderata for fast and data-oblivious delegated linear algebra. We demonstrate the existence of \textit{Trapdoored}-\textit{Matrix} families based on an LPN assumption, and provide a scheme for secure delegated matrix-matrix and matrix-vector multiplication based on the existence of trapdoored matrices. We achieve sublinear overhead for the server, dramatically reduced computation for the client, and various practical advantages over previous protocols.

    \keywords{Data Privacy, Data-Oblivious Computation, Delegation, Homomorphic Encryption, Cloud Computing, Sublinear Overhead, LPN, Matrix Multiplication, Trapdoored Matrix}
\end{abstract}

\section{Introduction}

Data-oblivious delegated computation -- computation in which the server does not learn anything meaningful about the input data, thanks to cryptographic scrambling by the client -- has been of theoretical and practical interest for some time \cite{pippenger1979relations,ostrovsky1990efficient}. Arbitrary computation is quite difficult in this setting, requiring Fully Homomorphic Encryption (FHE) \cite{rivest1978data,gentry2009fully} and currently incurring several orders of magnitude of overhead. As a result, much work on oblivious computation focuses on either advancing schemes for oblivious computation with more limited adversaries \cite{maas2013phantom,zahur2015obliv}, on developing and improving the runtime of FHE \cite{fan2012somewhat}, or on Partially Homomorphic Encryption (PHE) \cite{elgamal1985public,goldwasser2019probabilistic}, which focuses on encryption that commutes with \textit{some} operations.

A second paradigm, with different uses, has also emerged. Secure Multiparty Computation (MPC) \cite{yao1986generate,goldreich1987solve}, in which several non-colluding adversaries compute on a union of datasets private to each server, has allowed for a variety of recent developments in hidden computation protocols. Here, a combination of informational and cryptographic techniques often offer efficiency improvements over corresponding PHE problems~\cite{mohassel2017secureml,araki2016high,keller2020mp,knott2021crypten}.

Much work in FHE and MPC has focused on lattice techniques. Learning With Errors (LWE), a survey of which is available in \cite{regev2010learning}, is the most well-known cryptographic problem in this area. Broadly speaking, the problem asks a user to distinguish a sequence of random values from a sequence of low-magnitude-noised inner products of known vectors $a_i$ against some unknown vector $s$. Learning Parity with Noise (LPN), a variant problem, instead asks that the noise be sparse.

In recent decades, an increasing fraction of computational power has been dedicated to matrix-vector and matrix-matrix multiplications. Specialized hardware to perform these operations is expensive and power-hungry, and is therefore unusually centralized. We often wish to allow low-power devices (such as personal computing hardware) to perform large matrix-vector multiplications (in evaluating artificial neural networks, for instance) or to allow companies to rent computing power from centralized providers. However, input data must often be kept private (e.g. medical records, personal photos). We therefore desire schemes for delegated matrix-vector and matrix-matrix multiplication that are both secure and highly efficient -- when even insecure computation is quite expensive, a constant-factor overhead is extremely costly. Techniques using exclusively operations native to modern GPUs/TPUs are preferred, as they benefit from extensive hardware and software optimization.

\subsection{Our Contributions}

We wish to conduct delegated computations of the form $Av_i$ for some prespecified matrix $A$ and online vectors $v_i$, or matrix-matrix multiplications $AB$, while hiding $A, B, v_i$ from the server and minimizing client and total computational load. We propose a protocol for these problems using two novel cryptographic primitives: {\em trapdoored matrices}, matrices which are cryptographically pseudorandom but which can be multiplied by other matrices quickly using trapdoor information, and {\em targeted trapdoored matrices}, a slight generalization of trapdoored matrices that amortizes certain costs. At a high level, our protocol uses trapdoored matrices to additively mask operands before requesting their product from the server, and then uses the trapdoor property to quickly calculate the product of the unmasked operands from the product of the masked operands.

We then provide constructions of trapdoored objects, and an associated protocol, under a standard LPN hardness assumption. We also implement and benchmark our protocol. It is lightweight and easy to implement, secure over any LPN-friendly ring or field, and uses only standard matrix and vector operations, rendering it highly amenable to compiler and hardware acceleration. As a result, it achieves both theoretically and practically small complexity blowup over computation in the clear while substantially reducing client computation.

\subsection{Related Work}
\cite{sotiraki2016authentication}\footnote{We were unaware of this work in a previous version and apologize for the omission.} proposed a different version of trapdoored matrices, including a candidate construction, while studying identity verification. In this work, the definition of trapdoored matrices is not made exact, but the properties required are far weaker. Instead of requiring trapdoored matrices to be pseudorandom, this work merely requires that it be challenging to compute information enabling quick multiplication from inspection of a randomly sampled trapdoor matrix.

Several recent papers have focused directly on multiplication of matrices when one or both are encrypted, such as \cite{jiang2018secure,bae2024plaintext,liu2022privacy,gao2024secure}. \cite{liu2022privacy,bae2024plaintext} in particular achieve efficient computation and decoding in practical parameter regimes, but require expensive (as compared to computation) encoding steps. Many of these works also use LWE variants and incur a small accuracy error as a result.

LPN-type problems were first discussed in \cite{blum1993cryptographic} and considered over finite fields in \cite{ishai2009secure}. Algorithms for binary problems were considered in \cite{kearns1998efficient}, and best-known quantitative hardness results for LPN over rings are discussed in \cite{liu2024hardness}.

The link between oblivious linear function evaluation (OLE) and LPN was identified by \cite{ishai2009secure} (while we do not consider MPC or OLE directly, this line of work represents some of the more advanced uses of LPN to date). \cite{dottling2017tinyole,boyle2018compressing,boyle2019efficient,boyle2020efficient} further demonstrated the usefulness of LPN for secure linear computation in MPC. In particular, \cite{boyle2019efficient} presented a MPC protocol that can be turned into a vector-secure matrix-vector multiplication protocol, but focused on goals corresponding to the MPC setting. \cite{applebaum2017secure,applebaum2023actively} consider low-overhead MPC linear-algebra, including a protocol obtaining constant-factor overhead for secure two-party arithmetic computation with semi-honest adversaries. \cite{chen2020maliciously} used alternate techniques to achieve matrix multiplication with linear communication in an all-but-one-adversarial MPC framework.

Secure matrix-vector multiplication is useful in various contexts, including evaluation of neural networks \cite{mohassel2017secureml,dalskov2019secure,mann2023towards}. 

In a concurrent and independent work \cite{vaikuntanathan2025improving}, Vaikuntanathan and Zamir suggest a recursive construction of trapdoored matrices very similar to that of our Subsection~\ref{sec:improvedlpn} and use it to accelerate algorithms that use random matrices. Their construction can be used to slightly improve our server cost for matrix-matrix multiplication (removing the term in $\delta$) while approximately doubling client cost.

A very recent follow-up work \cite{benhamouda2025encrypted} has since provided alternate Trapdoored-Matrix and Targeted-Trapdoored-Matrix-like constructions using classical and novel cryptographic assumptions, enabling significant storage and performance improvements in some domains (see Section~\ref{sec:experiments} for discussion). We expect further improved Trapdoored-Matrix constructions, with associated improvements in linear algebra delegation, to be developed in the future.

\section{Fast Delegated Computation}
We work within the model of secure delegated computation problems. We assume two parties: a client with private data (e.g. two matrices, or a matrix and an online stream of vectors) and a server. The client wishes to learn some public function (e.g. matrix product) of their data while computationally hiding it from the server.

Since we consider clients that are much less computationally able than servers, we are interested in protocols which achieve the following properties:
\begin{itemize}
	\item Total protocol computation is very low (at most low-factor linear in, and ideally asymptotically equal to, the naive computational cost of the function).
	\item Total computation by the client is much less than that required to compute the function on its own (and ideally, nearly linear in input size).
	\item The protocol is computational zero-knowledge for dishonest servers.
    \item The client can detect dishonest server behavior.
    \item The protocol minimizes communication and round count.
\end{itemize}

For efficiency, we compare computational load to that of unencrypted delegation:
\begin{definition}
	We say that a client-server protocol in which the client receives data $A$ and is required to return prespecified $f(A)$ is $(c_c, c_s)$-efficient if the client and the server respectively expend $c_c$ and $c_s$ times their respective computational costs in the naive protocol in which the client sends $A$ to the server and the server computes and returns $f(A)$ to the client.
\end{definition}
$c_c$ and $c_s$ will typically depend on parameters (e.g. problem size, security parameter). An ideal protocol would have $c_c=c_s=1$, an efficient protocol has $c_c$ small and $c_s=1+o(1)$, and a protocol beats local computation if $c_c$ is less than the cost of computing the result locally divided by the total input-output size.

\section{Three Problems and One Solution}
\label{sec:mainthm}

We consider two particularly salient problems of linear algebra:
\begin{definition}[Matrix-Matrix]
	Given matrices $A\in R^{m\times n}, B\in R^{n\times l}$ over a ring $R$, return $AB$.
\end{definition}
\begin{definition}[Matrix-Vectors]
	Given a $m\times n$ matrix $M$ over a ring $R$ and an online stream $v_1, v_2, \dots\in R^n$, return $Mv_1, Mv_2, \dots$ online.
\end{definition}

Note that a secure delegation protocol for Matrix-Vectors yields a natural protocol for Matrix-Matrix (by multiplying one column of $B$ at a time), but that protocol does not enjoy the same efficiency as the Matrix-Vectors protocol, as the unencrypted complexity of Matrix-Matrix is lower than $n$-repeated Matrix-Vectors.

These jointly generalize to the following third problem:
\begin{definition}[Matrix-Matrices]
	Given $A\in R^{m\times n}$ and an online stream of matrices $B_i\in R^{n\times l}$ over a ring $R$, return the online stream $AB_i$.
\end{definition}

Our protocols for this problem have two phases: an initialization phase, in which the client and server interact for preprocessing purposes, and an online phase, in which the client and server interact after every online input.

Our results depend on an LPN hardness assumption (in Section~\ref{sec:lpn-assumption}), and have performance directly proportional to the strength of that assumption. To state performance concisely, let $\nu(\delta, \epsilon, \lambda)=\max\left(\frac{\lambda}{\lg |R|}, \iota\right)$, where $R$ is the ring over which we compute and $\iota$ is the smallest value such that LPN in dimension $\frac{\iota}{\delta}$ with subspace dimension $\iota$ and noise rate $\left(\frac{\iota}{\delta}\right)^{\epsilon-1}$ obtains $(\lambda+\lg \lg n+1)$-bit-security CZK. We then show the following:

\begin{theorem}
    \label{thm:full}
    For any ring $R$ over which the LPN assumption holds, Protocol~\ref{prot:improved} is a secure delegated Matrix-Matrices protocol. It guarantees, for any security parameter $\lambda$ and $\delta, \epsilon\in (0, 1)$ such that the LPN assumption holds:
    \begin{itemize}
        \item $\lambda$-bit-security CZK.
        \item Initialization costs \[O\left(n(m+n)\left(n^\epsilon + \nu(\delta, \epsilon, \lambda)^{\omega-2}\right)\right)\] and \[O\left(n^\omega \delta^{\omega-1} + n \max(m, \delta n)\min(m, \delta n)^{\omega-2}\right) \] to the client and server respectively.\footnote{Where $\omega$ is the matrix multiplication constant.}
        \item Per-round cost \[O\left((m+n)l\left(n^{\epsilon} + \nu(\delta, \epsilon, \lambda)\right)\right)\] for the client and per-round overhead (i.e. excluding the cost of a single naive-sized $m\times n$-by-$n\times l$ multiplication) \[O(n\max(\delta n, l) \min(\delta n, l)^{\omega-2})\] for the server.
        \item Two-round initialization and one-round online steps.
        \item Total communication of $1+O\left( \frac{\delta n}{m+l}\right)$ times input size.
    \end{itemize}
    In particular, in parameter regimes where $m=\Omega(n)$, the protocol is
    \[\left(n^\epsilon+\nu(\delta, \epsilon, \lambda), 1+O\left(\delta^{\omega-2}\right)\right)\]
    efficient in both Matrix-Matrix and $\max(n^\epsilon, (\delta n)^{\omega-2})$-steps-amortized Matrix-Vectors.
\end{theorem}

Two special cases are of particular interest.
\begin{corollary}
    Protocol~\ref{prot:improved} restricts to a two-round\footnote{A minor modification is required to parallelize independent server calls.} secure delegated $n$-by-$n$ matrix-matrix protocol with client cost $O(n^{2}(n^{\epsilon}+\nu(\delta, \epsilon, \lambda)))$, and server cost beyond a single $n$-by-$n$ matrix multiplication of $O\left(\delta^{\omega-2} n^\omega\right)$.
\end{corollary}
\begin{corollary}
    Protocol~\ref{prot:improved} restricts to a two-round-per-initialization, one-round-per-online-step secure delegated $n$-by-$n$ matrix-vectors protocol with:
    \begin{itemize}
        \item Client online cost $O(n(n^\epsilon+\nu(\delta, \epsilon, \lambda)))$
        \item Client initialization cost $O(n^{2}(n^\epsilon+\nu(\delta, \epsilon, \lambda)^{\omega-2}))$
        \item Server online cost beyond a single $n$-by-$n$ matrix-vector multiplication of $O\left(\delta n^2\right)$
        \item Server initialization cost of $O(\delta n^{\omega})$
    \end{itemize}
\end{corollary}

\section{Primitives for Fast Linear Algebra}
Our basic approach is a delegation-oriented modification of Beaver's triples \cite{beaver1992efficient}: to secure the computation $AB_i$, generate pseudorandom matrices $A'$ and $B'_i$. Request the computation $(A+A')(B_i+B'_i)$ from the server (sending the results of the sums and requesting their product), and return $(A+A')(B_i+B'_i)-A'B_i-(A+A')B'_i$ (where the latter two terms are computed locally). While this is clearly secure by the usual argument, it is not clearly efficient: it requires two local matrix multiplications. To improve efficiency, we must pick $A'$ and $B'_i$ to reduce the cost of these operations, inspiring the following cryptographic primitive generalizing an idea of \cite{sotiraki2016authentication}.

\begin{definition}[Trapdoored-Matrix]
    A Trapdoored-Matrix scheme is an algorithm that, given a ring $R$ and size $(m, n)$, generates a pseudorandom $m\times n$ matrix $M$ over $R$ and additional data such that, for any $N\in R^{n\times l}$, $MN$ can be computed more efficiently than via naive multiplication.\footnote{\cite{sotiraki2016authentication} proposes a weaker construction under a nonstandard security assumption. A construction of \cite{boyle2020efficient} (the first construction with the Toeplitz/quasi-cyclic assumption in Section 10.3) also fulfills these criteria, albeit with significantly stronger LPN assumptions and worse performance at equal security levels (as it was optimized for a substantially different task). A work concurrent to and independent of ours \cite{vaikuntanathan2025improving} and a work subsequent to ours \cite{benhamouda2025encrypted} propose other constructions.}
\end{definition}

This scheme suffices for matrix-matrix multiplication. However, we also wish to delegate matrix-vector multiplication, and therefore to develop pseudorandom trapdoored vectors whose product with an arbitrary matrix $A$ can be computed in time much less than the naive product. This is impossible for one-shot instances, as the time to compute $Av$ for arbitrary $v$ is asymptotically equal to the time required to read $A$; therefore, we consider a further generalization of the above.

\begin{definition}[Targeted-Trapdoored-Matrix]
    A Targeted-Trapdoored-Matrix scheme is an algorithm that, given a ring $R$, size $(m, n)$, $l\geq N((m, n, R))\in \Z$, and a matrix $M\in R^{n\times l}$, preprocesses $M$ and returns an online generator that generates pseudorandom $m\times n$ matrices $B_i$ and the associated values $B_iM$ in less time than is required to draw a uniformly random $B_i$ and compute the products $B_iM$.
\end{definition}

Combined, these make the template scheme for Matrix-Matrices fast even when $B_i$ is a vector: use the Trapdoored-Matrix generator to generate $A'$, and use the Targeted-Trapdoored-Matrix generator to generate the various $B'_i$.

At the cost of $O((n+m)tl)$ additional storage and of beginning the protocol with an $m\times n$-by-$n\times tl$ secured delegated matrix-matrix computation, we can obtain amortized Targeted-Trapdoored-Matrix for a $t$-matrix stream from Trapdoored-Matrix: generate a trapdoored matrix $M$, compute $AM$, and then use blocks of $M$ and $AM$ as the $B'_i$ and $AB'_i$, respectively.

We know of many families of matrices (e.g. Vandermonde, Cauchy, and discrete Chebyshev matrices) for which multiplication is $\tilde{O}(n)$. Since linear combinations of these or other matrices admitting fast multiplication could plausibly be pseudorandom, it seems natural to believe the following conjecture\footnote{An update: a subsequent work, \cite{benhamouda2025encrypted}, provides conjectured-secure Trapdoored-Matrix constructions with $O(n)$ multiplication cost.}:
\begin{conjecture}
\label{conj:hiddenFastMatrix}
	There exist Trapdoored-Matrix constructions with polynomial-time computational security, preprocessing cost $\tilde{O}(n^2)$ and matrix-vector multiplication cost $\tilde{O}(n)$.
\end{conjecture}

The complexity of Targeted-Trapdoored-Matrix in the case where $m=1$ (i.e. Targeted-Trapdoored-Vector) is less clear. The security assumption of Section \ref{sec:lpn-assumption} is commonly believed to hold for any $\delta,\epsilon>0$. Security of LPN in the dimension regime required for polylogarithmic overhead, however, is less plausible; if a polylogarithmic-overhead scheme exists, it will likely require another approach.

\section{An LPN Hardness Assumption and Performance Scaling}
\label{sec:lpn-assumption}
Given a ring $R$, a dimension $n$, a subspace dimension $\subspacedim$, and a noise rate $\noiserate$, the decisional LPN problem asks one to distinguish messages of the form $(L, LH+S)$ from messages of the form $(L, U)$ where $L\sampledfrom R^{n\times\subspacedim}$, $H\sampledfrom R^{\subspacedim\times \poly(n)}$, $U\sampledfrom R^{n\times \poly(n)}$, and $S\sampledfrom \mc{D}_{R, \noiserate}^{n\times \poly(n)}$ where $\mc{D}_{R, \noiserate}$ is the distribution that is 0 with probability $1-\noiserate$ and uniform over $R$ with probability $\noiserate$. A standard hardness assumption for LPN (as in \cite{boyle2018compressing,yu2021smoothing}) is the following:
\begin{conjecture}[Standard LPN Hardness Assumption]
    \label{conj:lpn-hardness}
    $\exists \epsilon<1$ such that fixing any $\delta>0$, there exists $n^*\in \N, c>0$ such that for any $n>n^*, \subspacedim=\delta n$, $\noiserate = n^{\epsilon-1}$, no polynomial-time adversary can distinguish the distribution of tuples $(L, LH+S)$ from that of $(L, U)$ with non-negligible advantage, where $L\sampledfrom R^{n\times\subspacedim}$, $H\sampledfrom R^{\subspacedim\times \poly(n)}$, $U\sampledfrom R^{n\times \poly(n)}$, and $S\sampledfrom \mc{D}_{R, \noiserate}^{n\times \poly(n)}$.
\end{conjecture}

Both $\delta$ and $\epsilon$ directly correspond to the parameters by the same names in Theorem \ref{thm:full}. In particular, slightly stronger assumptions directly translate into significant efficiency improvements. For instance, if $\delta$ is allowed to be $o(1)$, we achieve sub-multiplicative overhead for $\frac{m}{n}\geq \delta$, rather than just $m=\Omega(n)$. All terms of the form $n^\epsilon$ correspond to $n\mu$; if we assume that $\mu(n)=n^{\epsilon-1}$ is secure for every $\epsilon>0$, we obtain the corresponding substantial efficiency improvements.\footnote{Our LPN assumption is commonly believed to hold (nonuniformly) for all $\epsilon$, but our results are nontrivial for any $\epsilon<\omega-2$.} Likewise, the security of our protocol depends directly on the computational hardness assumption: a super-polynomial assumption would yield a protocol with super-polynomial security.

\section{A Simple Secure Delegated Matrix-Matrices Protocol}
We first use LPN to construct schemes for Trapdoored-Matrix and Targeted-Trapdoored-Matrix. Our approach is simple: we observe that a matrix of the form $LH+S$, where $L$ and $H$ have few columns and rows respectively and $S$ is sparse, is pseudorandom by the LPN assumption and is easy to multiply by other matrices. Moreover, given a matrix $A$, once $AL^\top$ has been computed, it is very efficient to compute $A(LH+S)$ as $(AL)H+AS$, even if $LH+S$ is a thin matrix (or vector). 

This results in a simple protocol for Secure Delegated Matrix-Matrices, detailed in Protocol~\ref{prot:simple}.

A brief note on notation: in order to aid the reader in tracking which parts of expressions have already been computed, we use angled brackets around an expression, such as $\inner{\mathrm{Expr}(x, y, z, \dots)}$, to denote an already-computed variable whose value is equal to that of the expression $\mathrm{Expr}(x, y, z, \dots)$.

\begin{protocol}{Simple Client Protocols for Matrix-Matrices}{simple}
\textbf{Inputs:} Ring $R$, left matrix $A\in R^{m\times n}$, subspace dimension $n_1$, sparsity $\noiserate$.
\bigskip

\textbf{Initialization:} 
\begin{stepenumerate}
    \item Sample LPN subspace matrix $L\sampledfrom R^{n\times n_1}$.
    \item Sample $H\sampledfrom R^{m\times n_1}, S\sampledfrom \mc{D}_{\noiserate}^{m\times n}$.
    \item Compute \[A'\coloneqq HL^\top+S.\]
    \item Compute and send $A_{\mathbf{enc}}\coloneqq A+A'$.
    \item Compute $A_{\mathbf{enc}}L, AL$.
\end{stepenumerate}
\bigskip
\textbf{Multiplication by $B\in R^{n\times l}$:}
\begin{stepenumerate}
    \item Sample $G\sampledfrom R^{n_1\times l}, T\sampledfrom \mc{D}_{\noiserate}^{n\times l}$.
    \item Compute \[B'\coloneqq LG + T.\] Compute and send $B_{\mathbf{enc}}\coloneqq B+B'$.
    \item Compute $AB'=\inner{AL}G + AT$, $L^\top B_{\mathbf{enc}}$
    \item Request $A_{\mathbf{enc}}B_{\mathbf{enc}}$.
    \item Compute \[A'B_{\mathbf{enc}}=H\inner{L^\top B_{\mathbf{enc}}} + SB_{\mathbf{enc}}.\]
    \item Return $AB=\inner{A_{\mathbf{enc}}B_{\mathbf{enc}}}-\inner{AB'}-\inner{A'B_{\mathbf{enc}}}$
\end{stepenumerate}
\end{protocol}

\subsection{Security and Performance}
The security of Protocol~\ref{prot:simple} is a direct consequence of the LPN assumption: the matrices $A'$ and $B'$ are constructed exactly as the matrix that the assumption holds to be pseudorandom conditioned on $L$, and the joint distribution of $(L, A+A', B+B')$ (from which all messages to the server can be derived) therefore leaks no information about $A$ and $B$.

Protocol~\ref{prot:simple} is moderately efficient, already allowing for client savings for large parameters. The initialization requires one $m\times \subspacedim$-by-$\subspacedim\times n$ matrix multiplication, one $m\times n$-by-$n\times\subspacedim$ matrix multiplication, and negligible additions. An online step has similar cost; in particular, letting $a, b, c$ be the ascending sort of $m, n, l$, and $a', b', c'$ be the ascending sort of $m, \subspacedim, l$, the cost for an online step is $O\left((a')^{\omega-2} b'c'+mnl\noiserate\right)$, while insecure delegation has cost $O\left((a)^{\omega-2} bc \right)$. For instance, in the case $m=n, l=1$ (corresponding to matrix-vectors), we drive the client cost down from $O(n^2)$ to $O(n(\subspacedim+n\noiserate))$. A server's online step merely consists of receiving two matrices of the problem sizes, computing the product, and returning it, so server computation time and round/bit efficiency of communication are optimal.

For most parameter regimes (e.g. $\delta=\Theta(1), \noiserate=o(1)$), the dominant cost comes from the dense multiplication $HL$, and therefore varies polynomially in subspace dimension $\subspacedim$. This can be greatly reduced by using LPN variants over fast matrices (e.g. sparse-LPN, Toeplitz-LPN, Quasi-cyclic-LPN \cite{aguilar2018efficient}), but precise parameter tradeoffs for these problems are not currently well-understood, and may entail worse $\mu$. In the next subsection, we demonstrate a more advanced protocol that obtains similar or better performance (depending on the comparative sparsity assumptions of LPN and sparse-LPN) than a sparse-LPN version, but depends only on standard LPN.

\section{An Improved Secure Delegated Matrix-Matrices Protocol}
\label{sec:improvedlpn}

Protocol~\ref{prot:simple} has client cost very roughly $(\mu+\delta)n$ times the cost of simply sending the inputs to the server. This is rather high. It is also not (just) an artifact of our weak parameter assumptions: the security of LPN does not degrade gracefully as its parameters are reduced. In particular, once $\delta\noiserate\leq \frac{1}{n}$, straightforward attacks from subsampling and Gaussian elimination begin to become efficient. This imposes a minimum multiplicative overhead of $\Theta(\sqrt{n})$ if super-polynomial security is desired. In this section, we present a more complicated protocol that bypasses this barrier to improve performance.

The client cost of Protocol~\ref{prot:simple} is dominated by the costs of the multiplications $AL$ and $(AL)G$. There is a natural solution: recursively call the server to compute these, while being careful to continue hiding $A$ and $G$. Additionally, computations where both matrices are public (for instance, those of the form $L(B+B')$) may be offloaded to the server. These changes result in a fairly complicated recursive protocol with several disadvantages, including a number of rounds growing exponentially in the recursion depth and substantial duplication of work, but this protocol in turn may be greatly optimized and simplified. We detail the resulting protocol for secure delegated matrix-matrices in Protocol~\ref{prot:improved} (and a non-interactive verification subroutine based on Freivalds' algorithm \cite{freivalds1979fast} in Protocol~\ref{prot:checkpartial}), and spend the rest of this section on exposition and showing that the protocol fulfills Theorem~\ref{thm:full}.

\begin{protocol}{Checking Partial Products}{checkpartial}
    \textbf{Inputs:} Ring $R$, matrices $M_i\in R^{a_i\times a_{i+1}}\forall i\in [d]$, matrices $P_i\in R^{a_1\times a_{i+1}}\forall i\in \{2, 3, \dots, d\}$, number of product checks $\lambda'$.
    \begin{stepenumerate}
        \item Let $P_1=M_1$.
        \item For each $i\in \{2, \dots, d\}$
        \begin{stepenumerate}
            \item Sample $X\sampledfrom R^{a_{i+1}\times \lambda'}$.
            \item If $P_i X\neq P_{i-1}(M_iX)$, return \textbf{False}.
        \end{stepenumerate}
        \item Return \textbf{True}.
    \end{stepenumerate}
\end{protocol}

\begin{protocol}{Improved Client Protocols for Matrix-Matrices}{improved}
\textbf{Inputs:} Ring $R$, left matrix $A\in R^{m\times n}$, subspace dimensions $(n=n_0), n_1, \dots, n_d$, sparsities $\mu_1, \dots, \mu_d$, number of product checks $\lambda'$.
\medskip

\textbf{Initialization:} 
\begin{stepenumerate}
    \item Sample and send LPN subspace matrices $L_i\sampledfrom R^{n_{i-1}\times n_i}\forall i\in [d]$.
    \item For each pair $(i, j)\in \{0, 1, \dots, d\}^2$, request the partial product $L_j^\top L_{j-1}^\top \dots L_1^\top L_1 L_2 \dots L_i$. Check these via Protocol~\ref{prot:checkpartial}, and abort if any are incorrect.
    \item Sample $H\sampledfrom R^{m\times n_d}, S_i\sampledfrom \mc{D}_{\noiserate_i}^{m\times n_{i-1}}\forall i\in [d]$.
    \item Compute \[A'\coloneqq H\inner{L_1L_2\dots L_d}^\top + \sum_{i=0}^{d-1} S_{i+1}\inner{L_1L_2\dots L_i}^\top .\]
    \item Compute and send $A_{\mathbf{enc}}\coloneqq A+A'$.
    \item Request and check via Protocol~\ref{prot:checkpartial} the partial products $A_{\mathbf{enc}}L_1$, $A_{\mathbf{enc}}L_1L_2$, \dots, $A_{\mathbf{enc}}L_1L_2\dots L_d$. Abort if any are incorrect.
    \item Compute, for every $i\in [d]$,
    \begin{align*}
        A'L_1L_2\dots L_i&=H\inner{L_d^\top L_{d-1}^\top \dots L_1^\top L_1L_2\dots L_i}\\
        &\quad + \sum_{j=0}^{d-1} S_{j+1}\inner{L_j^\top L_{j-1}^\top \dots L_1^\top L_1L_2\dots L_i}\\
        AL_1L_2\dots L_i &= \inner{A_{\mathbf{enc}}L_1L_2\dots L_i} - \inner{A'L_1L_2\dots L_i}.
    \end{align*}
\end{stepenumerate}
\medskip
\textbf{Multiplication by $B\in R^{n\times l}$:}
\begin{stepenumerate}
    \item Sample $G\sampledfrom R^{n_d\times l}, T_i\sampledfrom \mc{D}_{\noiserate_i}^{n_{i-1}\times l}\forall i\in [d]$.
    \item Compute \[B'\coloneqq \inner{L_1L_2\dots L_d}G + \sum_{i=0}^{d-1} \inner{L_1L_2\dots L_i} T_{i+1}.\] Compute and send $B_{\mathbf{enc}}\coloneqq B+B'$.
    \item Compute $AB'=\inner{AL_1L_2\dots L_d}G + \sum_{i=0}^{d-1} \inner{AL_1L_2\dots L_i} T_{i+1}$.
    \item Request $A_{\mathbf{enc}}B_{\mathbf{enc}}$ and partial products $(L_1)^\top B_{\mathbf{enc}}$, $(L_1L_2)^\top B_{\mathbf{enc}}$, \dots, $(L_1L_2\dots L_d)^\top B_{\mathbf{enc}}$.
    \item Compute \[A'B_{\mathbf{enc}}=H\inner{(L_1L_2\dots L_d)^\top B_{\mathbf{enc}}} + \sum_{i=0}^{d-1} S_{i+1}  \inner{(L_1L_2\dots L_i)^\top B_{\mathbf{enc}}}.\]
    \item Return $AB=\inner{A_{\mathbf{enc}}B_{\mathbf{enc}}}-\inner{AB'}-\inner{A'B_{\mathbf{enc}}}$
\end{stepenumerate}
\end{protocol}

\subsection{Protocol Structure}

The structure of Protocol~\ref{prot:improved} closely parallels that of Protocol~\ref{prot:simple}. The key difference is that pseudorandom matrices are now constructed via iteration of the LPN construction; for instance,
\begin{align*}
    B' &\coloneqq L_1L_2\dots L_d G + \sum_{i=0}^{d-1} L_1L_2\dots L_{i} T_{i+1}\\
    &=(T_1 + L_1(T_2+L_2(T_3+L_3(\dots(T_d+L_dG))))).
\end{align*}

This change makes the client-side dense multiplications substantially smaller, improving efficiency. It does, however, require precomputation of various products of the LPN subspace matrices and $A$. Luckily, these products can themselves be offloaded to the server (including by a partially recursive masking trick, as seen on line 7), and owing to the thinness of the subspace matrices, have low cost relative to the target multiplications $AB_i$. The client must check the honesty of the server on these multiplications; Protocol~\ref{prot:checkpartial}, a slight generalization of Freivalds' algorithm \cite{freivalds1979fast}, does so efficiently.

\subsection{Security and Correctness Guarantees}

Our protocol makes three guarantees. First, in the semi-honest model, the protocol will return exactly $AB$. Second, the protocol is computational zero-knowledge against even a dishonest server. Third, we obtain a weaker form of correctness of outputs: a server that behaves dishonestly in any setup query or fraction $\geq \alpha$ of the online invocations can be detected with high probability with $O\left(\frac{1}{\alpha}\right)$ additional queries (and if $l$ and $|R|$ are large, much more efficient per-call checks are available via Freivalds' algorithm).

The first guarantee follows from expanding the return line and noting that it is equal to $AB$. We emphasize that all LPN noise \textit{exactly} cancels; unlike some LWE-based protocols, we introduce no error in the output.

For the security guarantees, we first notice that Freivalds' multiplication check (and therefore Protocol~\ref{prot:checkpartial}'s iterated version) will return true with probability 1 if the products are correct and probability $|R|^{-\lambda'}$ otherwise, where $\lambda'$ is the number of trials performed. Note that this provides low soundness error at very low computational cost for large rings: for 32-bit integers, for instance, we achieve soundness error $2^{-128}$ when $\lambda'=4$. Assuming that $\lambda'$ is picked so that this probability is negligible, we demonstrate that Protocol~\ref{prot:improved} is CZK.

\begin{theorem}
    Given the LPN security assumption of Section \ref{sec:lpn-assumption}, Protocol~\ref{prot:improved} is computational zero-knowledge for the server.
\end{theorem}
\begin{proof}
    We construct a server-view simulator. The simulator is quite simple: it runs a copy of the client with inputs $A=0, B_i=0$. The $L_i$ generated by the simulator are identically distributed to those generated by the true client. If the server misreports any of the products $L_j^\top L_{j-1}^\top \dots L_1^\top L_1L_2\dots L_i$ or $A_{\mathbf{enc}}L_1L_2\dots L_i$, both the simulator and a true client will detect the error and terminate the protocol with probability $\geq 1-|R|^{-\lambda'}$, so the complete views in this case are statistically indistinguishable. 
    
    Otherwise, letting $G^{(j)}, S_i^{(j)}$ be the particular $G, S_i$ used in generation of $B'$ in the $j$th online step, the rest of the client's messages to the server can be concatenated into the matrix
    \begin{align*}
        D + \inner{L_1\dots L_d} X + \sum_{i=0}^{d-1} \inner{L_1\dots L_i} Y_{i+1}
    \end{align*}
    where 
    \begin{align*}
        D&=\begin{bmatrix}
            A^\top&B^{(1)}&B^{(2)}&\dots &B^{(t)}
        \end{bmatrix}\\
        X&=\begin{bmatrix}
            H^\top & G^{(1)}& G^{(2)}& \dots&G^{(t)}
        \end{bmatrix}\\
        Y_i&=\begin{bmatrix}
            S_i^\top& T^{(1)}_i &T^{(2)}_i& \dots& T^{(t)}_i
        \end{bmatrix}
    \end{align*}
    Note in particular that $X\sampledfrom R^{n_d\times \poly(n)}, Y_i\sampledfrom \mc{D}_{\mu_i}^{n_{i-1}\times \poly(n)}$. Then the server view is contained in the tuple 
    \begin{align*}
        \left((L_i)_{i\in [d]}, D + \inner{L_1\dots L_d} X + \sum_{i=0}^{d-1} \inner{L_1\dots L_i} Y_{i+1}\right)
    \end{align*}
    
    It therefore suffices to show that, for any $D, D'$, the distribution (over $L_i, X, Y$) of tuples of the form above is computationally indistinguishable from that of those of the form
    \begin{align*}
        \left((L_i)_{i\in [d]}, D'+\inner{L_1\dots L_d} X + \sum_{i=0}^{d-1} \inner{L_1\dots L_i} Y_{i+1}\right)
    \end{align*}
    (i.e. where we replace $D$ with $D'$) for any $D, D'$. The following lemma shows that any such distribution is indistinguishable from the distribution resulting from replacing the second tuple element with independent and uniform noise, and thus from that generated by any other $D'$, completing the proof.
\end{proof}

\begin{lemma}
    Fix any $D$. For $j\in \{1, \dots, d\}$, let
    \begin{align*}   
        \mc{E}_j&=\left\{\left(\left(L_i\right)_{i\in [d]},  D+ L_1L_2\dots L_j X_j + \sum_{i=0}^{j-1} (L_1 L_2 \dots L_{i})Y_{i+1}\right) \middle| \right. \\
        &\left.\phantom{\Bigg(} L_i\sampledfrom R^{n_{i-1}\times n_{i}},X_j\sampledfrom R^{n_j\times \poly(n)}, Y_i\leftarrow \mc{D}^{n_{i-1}\times \poly(n)}_{\mu_i}\right\}
    \end{align*}

    Let 
    \begin{align*}
        \mc{E}_0&=\left\{\left(\left(L_i\right)_{i\in [d]},  D+U\right) \middle|L_i\sampledfrom R^{n_{i-1}\times n_{i}}, U \sampledfrom R^{n_0\times \poly(n)}\right\}\\
        &=\left\{\left(\left(L_i\right)_{i\in [d]},  U\right) \middle|L_i\sampledfrom R^{n_{i-1}\times n_{i}}, U \sampledfrom R^{n_0\times \poly(n)}\right\}
    \end{align*}

    Under the assumption of Section~\ref{sec:lpn-assumption} and assuming $n_i\geq \nu(\delta, \epsilon, \lambda), n_i\geq \delta n_{i-1}, \mu_i\geq n_{i-1}^{\epsilon-1} \forall i\in [d]$, $\lambda=\Omega(\lg n)$, $\mc{E}_d$ is $(\lambda-\lg d)$-bit computationally indistinguishable from the noise distribution $\mc{E}_0$.
\end{lemma}
\begin{proof}
    We prove the contrapositive. Let there be an efficient algorithm $\mc{A}$ distinguishing $\mc{E}_0$ and $\mc{E}_d$ with advantage $\geq \sum_{i=1}^d \epsilon_i$. Then there exists $j\in [d]$ s.t. $\mc{E}_{j-1}$ and $\mc{E}_{j}$ are distinguished by $\mc{A}$ with advantage $\geq \epsilon_j$. Then there exists an efficient algorithm distinguishing the distributions 
    \begin{align*}
        \left\{(L_j, L_j X_j + Y_j)\middle|L_j\sampledfrom R^{n_{j-1}\times n_{j}}, X_j\sampledfrom R^{n_j\times \poly(n)}, Y_j\leftarrow \mc{D}^{n_{j-1}\times \poly(n)}_{\mu_j}\right\}
    \end{align*}
    and 
    \begin{align*}
        \left\{(L_j, U)\middle|L_j\sampledfrom R^{n_{j-1}\times n_{j}}, U\sampledfrom R^{n_{j-1}\times \poly(n)}\right\}
    \end{align*}
    with advantage $\epsilon_j$: enhance samples from the former into samples from $\mc{E}_j$ and samples from the latter into samples from $\mc{E}_{j-1}$ by drawing the other random variables and computing the second tuple element from the results, and then apply $\mc{A}$. This shows that the LPN assumption does not hold.
\end{proof}

For the final security guarantee, we observe that the only place that the server can deviate from protocol without high-probability detection is in line 4 of the online multiplication subroutine (as all other requests are checked by Protocol~\ref{prot:checkpartial}). To check for an $\alpha$-fraction of deviations in these queries, the client may introduce $\frac{c}{\alpha}$ multiplication queries with $B_i=0$. By linearity of the downstream sections of the protocol, the error introduced in either query propagates linearly into the result through $H$ and the $S_i$, and so if the server behaves dishonestly, the client will compute an incorrect result with high probability. The client may check this -- $AB_i=0$ if $B_i=0$ -- and given deviation rate $\alpha$, the probability of a deviation in at least one of the $\frac{c}{\alpha}$ checking queries is $\geq 1-e^{-c}$, so any significant error rate can be detected cheaply and with high probability.

\subsection{Performance}
See Appendix~\ref{sec:performance} for asymptotic and per-operation performance exposition.

\subsection{Practical Improvements}
\label{subsec:practical}

We proved bounds in the case where $\delta, \mu$ are constant across layers. This is inefficient; by shrinking $\delta, \mu$ in larger layers to get approximately uniform security, dramatic improvements (in the server's $o(1)$ term and the client's full cost) can be achieved.

Careful readers may note that this protocol's pseudorandom generation is very similar to a non-recursive instance of dual-LPN on the block syndrome matrix
\begin{align*}
    \left[L_1, L_1L_2, \dots, L_1L_2\dots L_d \right]
\end{align*}

It may initially appear that we can simply substitute dual-LPN for this entire process and gain major client performance and simplicity improvements. We can, but only at constant-factor (over unencrypted computation) cost to the server, as the server would need to multiply the syndrome matrix by $A_{\mathrm{enc}}, B_{\mathrm{enc}}$. We can, however, achieve nearly the same improvements by substituting dual-LPN for the \textit{final} recursive step (i.e. replacing the generation $L_dG+T_d$ by $PT'_d$). This entirely eliminates the client dense multiplication, which is quite costly in regimes where $n$ is small. This does increase the server overhead, but this increase is on the multiplicative order of $\delta^{d-1}$. As a result, client performance improves substantially (and server performance decreases marginally) when $n$ is close to $n_d$, which occurs when $n$ is low or particularly high security is desired.

\section{Experiments}
\label{sec:experiments}
We implemented Protocol~\ref{prot:improved} in C++ and benchmarked it for both matrix-matrix multiplication and repeated matrix-vector multiplication. Experiments ran on a single core of an AMD Ryzen 5800X3D processor. Parameters were chosen for 128-bit security, using updated versions of the estimates of \cite{liu2024hardness}. Results are in Tables~\ref{fig:matrix-vector-plus1} and~\ref{fig:matrix-matrix-plus1}. Client R. and Total R. denote the ratios of client and total compute time, respectively, to the unencrypted compute time (denoted Local). We noticed substantially worse performance when $n$ was a power of two, presumably due to cache associativity; therefore, we set $n$ to be one more than a power of two. Benchmarks with $n$ exactly a power of two are available in Tables~\ref{fig:matrix-vector} and~\ref{fig:matrix-matrix} of the appendix; global performance is substantially worse, but relative performance is fairly consistent.

\begin{table}
    \centering
    \begin{tabular}{|c||c|c|c|c|c|c|c|}
        \hline
        $n$ & Local (s) & Client Init. (s) & Server Init. (s) & Client (s) & Server (s) & Client R. & Total R.\\
        \hline\hline
        1025 & 0.00113 & 0.0000729 & 0.0000371 & 0.00117 & 0.00121 & 1.09 & 2.19 \\
        \hline
        2049 & 0.00520 & 0.000145 & 0.0000830 & 0.00306 & 0.00527 & 0.616 & 1.65 \\
        \hline
        4097 & 0.0570 & 0.000415 & 0.000388 & 0.00605 & 0.0576 & 0.113 & 1.13 \\
        \hline
        8193 & 0.268 & 0.000927 & 0.00173 & 0.0145 & 0.270 & 0.0573 & 1.07 \\
        \hline
        16385 & 1.27 & 0.00221 & 0.00467 & 0.0475 & 1.27 & 0.0392 & 1.04 \\
        \hline
    \end{tabular}
    \caption{Matrix-Vectors performance benchmarks. Client and server times separated into amortized initialization and per-online-step. Average of five trials. Initialization times amortized across $n$ operations. Note that some server time is conservatively attributed to the client in simulation, so server times are slightly below naive time for small operand sizes.}
    \label{fig:matrix-vector-plus1}
\end{table}

\begin{table}
    \centering
    \begin{tabular}{|c||c|c|c|c|c|}
        \hline
        $n$ & Local (s) & Client (s) & Server (s) & Client R. & Total R.\\
        \hline\hline
        1025 & 0.0751 & 0.232 & 0.121 & 3.08 & 4.69 \\
        \hline
        2049 & 0.514 & 1.11 & 0.775 & 2.16 & 3.67 \\
        \hline
        4097 & 4.46 & 5.79 & 6.15 & 1.30 & 2.68 \\
        \hline
        8193 & 61.9 & 36.7 & 95.4 & 0.592 & 2.13 \\
        \hline
        16385 & 511 & 200 & 675 & 0.392 & 1.71 \\
        \hline
    \end{tabular}
    \\[0.5em]  
    \caption{Matrix-Matrix performance benchmarks. Average of five trials.}
    \label{fig:matrix-matrix-plus1}
\end{table}

We see substantial performance improvements as $n$ grows in both contexts, but especially in matrix-vectors. Here, we achieve client improvements even at small $n$, with improvements increasing to factor-$25$ speedup for the client as we grow the matrix. Initialization costs are consistently negligible (due to the advantage of matrix-matrix over repeated matrix-vectors). In matrix-matrix, client improvements are less substantial, in large part due to the locality efficiencies available in dense (but not sparse) matrix multiplication. However, clients still achieve substantial advantage for medium-size matrices, up to nearly factor-three improvement for $n=2^{14}$.

To the best of our knowledge, our method is the first to provide practical efficiency gains for the client in this parameter regime when encryption time is taken into account. Our efficiency is over two orders of magnitude higher than that of \cite{bae2024plaintext}, despite encrypting both multiplicands while their approach encrypts one (note also that almost all of their cost is incurred by encryption/decryption).\footnote{Their work, in turn, significantly outperforms \cite{liu2022privacy}.} Our protocol may be straightforwardly modified to only encrypt one multiplicand; this results in a significant additional improvement in client performance and server overhead.

Performance comparison to the later work of \cite{benhamouda2025encrypted} is more complex.\footnote{Our initial draft did not contain experimental data at their time of publication, so their work was unable to include experimental comparisons.} We work over $\Z/2^{32}\Z$, while \cite{benhamouda2025encrypted} works over $\F_p$ for $p$ a $32$-bit prime (and compares to plaintext operations over that field), rendering relative-to-plaintext performance incomparable. In addition, results are highly sensitive to deviations in exact security assumptions. With these caveats, \cite{benhamouda2025encrypted} obtains significant performance improvements for smaller values of $n$ ($\leq 2048$), consistent with the increasing importance of some of their overhead improvements at those scales, while our methods are more efficient for larger $n$. Their methods also offer significant savings on client-side storage for matrix-vectors. The existence of constructions with tradeoffs suggests potential for future improvement.

We do not know of an efficient and well-motivated method for parameter search for optimal subspace dimensions/sparsities. Instead, we use a heuristic: at each step, aim for a constant ratio of subspace dimension to sparsity. We performed (fairly sparse) grid search for the optimal ratio for each $n$; therefore, our results merely upper-bound optimal runtimes.

\section{Conclusion and Further Questions}
We present natural definitions of Trapdoored-Matrix and Targeted-Trapdoored-Matrix objects, a method to turn instances of these objects into protocols for secure delegated linear algebra, and optimized constructions of all of the above. More recently, other constructions of the above primitives have been put forward, and the existence of still better constructions remains likely. Variants of our schemes may also be extensible to more general problems; harder problems in linear algebra, such as matrix polynomial evaluation, are a reasonable next step. Achieving a variety of operations comparable to those of \cite{boyle2020correlated,couteau2023pseudorandom,bae2024plaintext} (while maintaining efficiency gains) seems significantly more challenging, but would yield even broader applications.

\begin{credits}
\subsubsection{\ackname}
Research supported in part by the NSF Alan T. Waterman Award, Grant No. 1933331.

Thanks to Alex Lombardi for helpful discussion on LPN and related topics. Thanks to the TCC'25 reviewers for helpful comments and suggestions.
\end{credits}

\bibliographystyle{splncs04}  
\bibliography{main}

\appendix

\section{Performance Analysis of Protocol~\ref{prot:improved}}
\label{sec:performance}
\subsection{Asymptotic Analysis}
We present a high-level exposition of performance here. An analysis of exact operation counts is available in Appendix~\ref{sec:detailedperformance}.

Fix $\delta<\frac{1}{2}, \epsilon, \lambda$. Let $\nu$ be as stated in Section~\ref{sec:mainthm}. We assume $m, n \geq \nu$. Set $n_i=\delta n_{i-1}$, $\mu_i = n_{i-1}^{\epsilon-1}$, $n_d=\nu$. We have $\sum_{i=1}^d n_i\leq \frac{\delta n}{1-\delta}$, $\mu_i n_{i-1}\leq n_{i-1}^\epsilon$. We set $\lambda'=\frac{\lambda}{\lg |R|}$ to achieve statistical $\lambda$-bit security in multiplication checks.

For simplicity, we omit the costs of the many calls to Protocol~\ref{prot:checkpartial}, as they are never dominant in reasonable parameter regimes. They are included in Appendix~\ref{sec:detailedperformance}, but it suffices to observe that they are dominated by $O(1)$ multiplications of $O(m+n)\times O(n)$ matrices by $O(n)\times O(\lambda')$ matrices.

\paragraph{Initialization:} The cost to compute all product sequences of the $L_i$s is $O(n^\omega \delta^{\omega - 1})$ to the server (as this step is dominated up to constants by a $\delta n\times n$-by-$n\times \delta n$ multiplication, with cost $\frac{1}{\delta} (\delta n)^\omega$). The cost to compute $A'$ is $O(mn \nu^{\omega - 2})$ for the first operation and $O(mn\sum_{i=0}^d (n \delta^i)^\epsilon)=O(mn^{1+\epsilon})$ for the second. The cost to compute the partial products with $A_{\mathbf{enc}}$ is $O\left(\min(m, \delta n)^{\omega - 2} n\max(m, \delta n)\right)$. The cost to compute the partial products of $A'$ with the various $L_i$ is dominated by the cost to compute $A'$ itself. Total cost (omitting check terms) is therefore $O(n^\omega \delta^{\omega-1} + \min(m, \delta n)^{\omega - 2} n\max(m, \delta n))$ to the server and $O(n(m+n)(n^\epsilon + \nu^{\omega - 2}))$ to the client.

\paragraph{Multiplication:} Computation of $B', AB'$ costs $O(l(m+n) (\nu+n^\epsilon))$ to the client (we do not get $\nu^{\omega-2}$ as $l$ may be smaller than $\nu$; the bounds can be improved by splitting on this event). Computation of $A_{\mathbf{enc}}B_{\mathbf{enc}}$ costs $O(mnl \min(m, n, l)^{\omega-3})$, and partial products cost $O(n\max(l, \delta n) \min(l, \delta_n)^{\omega-2})$ to the server. Computation of $A'B_{\mathbf{enc}}$ costs $O(lm(\nu+n^\epsilon))$ to the client. Then the total cost is, apart from the large multiplication, $O(n\max(l, \delta n) \min(l, \delta n)^{\omega-2})$ to the server and $O(l(m+n) (\nu+n^\epsilon))$ to the client.

We expound on two special cases:
\subsubsection{Matrix Multiplication}
One natural case is $m=n=l$ and one online step (i.e. a once-off matrix-matrix multiplication). Here, total client cost is $O\left(n^2\left(n^\epsilon+\nu^{\omega-2}\right)\right)$ and total server cost is, apart from one $n\times n$-by-$n\times n$ matrix multiplication, $O(n^\omega \delta^{\omega-2})$. In particular, our protocol is $\left(n^\epsilon+\nu^{\omega-2}, 1+O(\delta^{\omega-2})\right)$-efficient, resulting in a multiplicative $n^{1-\epsilon}$ saving for the client (ignoring the additive term in security parameter) and small overhead for the server.

\subsubsection{Online Matrix-Vectors}
Another natural case is $m=n, l=1$, and $k$ online steps. Here, setup cost is the same $O\left( n^2\left( n^\epsilon+\nu^{\omega-2}\right)\right)$ for the client and $O\left(\delta n^\omega\right)$ for the server, which amortize across $n^\epsilon$ and $\delta n$ online rounds, respectively. The per-multiplication online cost is $O(n(n^{\epsilon} + \nu))$ for the client and, apart from one $n\times n$ matrix-vector multiplication, $O(\delta n^2)$ for the server. We therefore obtain (after amortization) the same $\left(n^\epsilon+\nu, 1+O(\delta^{\omega-2})\right)$-efficiency guarantee and small overheads. Note that the amortization occurs in a sublinear number of online steps; this holds because the initial multiplication benefits from fast matrix multiplication, while the online steps do not.

\subsection{Per-Operation Analysis}
\label{sec:detailedperformance}

We analyze the performance componentwise and with precise constants, as our protocol exhibits particularly low constants. For simplicity, we omit steps with negligible performance impact.

Let $\sparsedotcost(a, x)$ be the cost of multiplying a $a\times ?$ matrix by a length-$?$, $x$-sparse vector (considered to be independent of the $?$) and $\matmulcost(a, b, c)$ be the cost of multiplying two matrices, one of which is $a\times b$ and the other of which is $b\times c$. All operations are over our working ring $R$.

\paragraph{{Protocol~\ref{prot:checkpartial}:}} Costs $\sum_{i=2}^{d} \matmulcost(a_1, a_i, \lambda') + \matmulcost(a_i, a_{i+1}, \lambda') + \matmulcost(a_1, a_{i+1}, \lambda')$, plus equality checks. Note that the third multiplication can be blocked with either of the two previous.

\paragraph{{Protocol~\ref{prot:improved}, Initialization line 2:}} One way to order these calculations is to, for every $i\in [d]$, calculate \[L_j^\top L_{j-1}^\top \dots L_1^\top L_1 L_2\dots L_{i}=\inner{L_j^\top L_{j-1}^\top \dots L_1^\top L_1 L_2\dots L_{i-1}} L_i,\] for every $j\in [i-1]$ and then calculate \[L_i^\top L_{j-1}^\top \dots L_1^\top L_1 L_2\dots L_{i}=\inner{L_{i-1}^\top L_{j-1}^\top \dots L_1^\top L_1 L_2\dots L_{i}}^\top L_i.\]

The calculations in the inner loop can be batched. This gives cost \[\sum_{i=1}^d \left[\matmulcost\left(\sum_{j=1}^{i-1} n_j, n_{i-1}, n_i\right) + \matmulcost\left(n_i, n_{i-1}, n_i\right)\right]\] for the server, and 
\begin{align*}
    \sum_{i=1}^d &\left[\matmulcost\left(\sum_{j=1}^{i-1} n_j, n_{i-1}+n_i, \lambda'\right)+2\matmulcost\left(n_{i-1}, n_i, \lambda'\right) + \right] +\matmulcost\left(n_{i},n_i+n_{i-1}, \lambda'\right)\\
\end{align*}
for the client checks.

\paragraph{{Protocol~\ref{prot:improved}, Initialization line 4:}} This costs $\matmulcost(m, n_d, n) + m \sum_{i=0}^{d-1} \sparsedotcost(n, n_{i}\mu_{i+1})$.

\paragraph{{Protocol~\ref{prot:improved}, Initialization line 6:}} This costs the server $\sum_{i=1}^d \matmulcost(m, n_{i-1}, n_i)$, and checking costs the client $\sum_{i=1}^d \matmulcost(m, n_{i-1}+n_i, \lambda') + \matmulcost(n_{i-1}, n_i, \lambda')$.

\paragraph{{Protocol~\ref{prot:improved}, Initialization line 7:}} This costs 
\begin{align*}
    \sum_{i=1}^d \matmulcost(m, n_d, n_i) + \sum_{j=0}^{d-1} m\sparsedotcost(n_i, \mu_{j+1} n_j)
\end{align*}

\paragraph{{Protocol~\ref{prot:improved}, Multiplication line 2:}} This costs $\matmulcost(l, n_d, n) + l \sum_{i=0}^{d-1} \sparsedotcost(n, n_{i}\mu_{i+1})$.

\paragraph{{Protocol~\ref{prot:improved}, Multiplication line 3:}} This costs $\matmulcost(l, n_d, m) + l \sum_{i=0}^{d-1} \sparsedotcost(m, n_{i}\mu_{i+1})$.

\paragraph{{Protocol~\ref{prot:improved}, Multiplication line 4:}} This costs the server $\matmulcost(m, n, l)$ for the first product (note that this is exactly the naive multiplication cost, and is the only term of this order), and $\sum_{i=1}^d \matmulcost(l, n_{i-1}, n_i)$ for the rest by computing partial products.

\paragraph{{Protocol~\ref{prot:improved}, Multiplication line 5:}} This costs $\matmulcost(m, n_d, l) + m \sum_{i=0}^{d-1} \sparsedotcost(l, n_{i}\mu_{i+1})$.

\section{Benchmarks for power-of-two $n$}

\begin{table}
    \centering
    \begin{tabular}{|c||c|c|c|c|c|c|c|}
        \hline
        $n$ & Local (s) & Client Init. (s) & Server Init. (s) & Client (s) & Server (s) & Client R. & Total R.\\
        \hline\hline
        1024 & 0.00523 & 0.0000746 & 0.0000315 & 0.00396 & 0.00472 & 0.773 & 1.68 \\
        \hline
        2048 & 0.0239 & 0.000155 & 0.0000743 & 0.0121 & 0.0238 & 0.513 & 1.51 \\
        \hline
        4096 & 0.104 & 0.000420 & 0.000402 & 0.0138 & 0.106 & 0.137 & 1.16 \\
        \hline
        8192 & 0.422 & 0.000932 & 0.00172 & 0.0300 & 0.429 & 0.0732 & 1.09 \\
        \hline
        16384 & 1.65 & 0.00279 & 0.00849 & 0.0826 & 1.66 & 0.0519 & 1.06 \\
        \hline
    \end{tabular}
    \caption{Variant of Table~\ref{fig:matrix-vector-plus1} without offset $n$. Note that some server time is conservatively attributed to the client in simulation, so server times are slightly below naive time for small operand sizes.}
    \label{fig:matrix-vector}
\end{table}

\begin{table}
    \centering
    \begin{tabular}{|c||c|c|c|c|c|}
        \hline
        $n$ & Local (s) & Client (s) & Server (s) & Client R. & Total R.\\
        \hline\hline
        1024 & 0.0716 & 0.239 & 0.112 & 3.33 & 4.90 \\
        \hline
        2048 & 0.467 & 1.10 & 0.720 & 2.35 & 3.89 \\
        \hline
        4096 & 5.22 & 6.03 & 5.96 & 1.15 & 2.30 \\
        \hline
        8192 & 60.0 & 35.5 & 88.0 & 0.592 & 2.06 \\
        \hline
        16384 & 503 & 190 & 812 & 0.378 & 1.99 \\
        \hline
    \end{tabular}
    \\[0.5em]  
    \caption{Variant of Table~\ref{fig:matrix-matrix-plus1} without offset $n$.}
    \label{fig:matrix-matrix}
\end{table}

\end{document}